\documentclass[sigconf]{acmart}

\usepackage{booktabs} 
\usepackage{multirow}
\usepackage{array}
\newcolumntype{M}[1]{>{\centering\arraybackslash}m{#1}}
\newcolumntype{P}[1]{>{\centering\arraybackslash}p{#1}}

\newtheoremstyle{scsthe}
{8pt}
{8pt}
{\it}
{}
{\bf}
{.}
{.5em}
{}

\theoremstyle{scsthe}
\newtheorem{theorem}{Theorem}

\newtheorem{definition}{Definition}

\newtheorem{lemma}{Lemma}[]

\usepackage{supertabular,booktabs}
\usepackage{tikz-cd}

\setcopyright{acmlicensed}

\acmDOI{10.1145/3314058.3314065}

\acmISBN{978-1-4503-7147-6/19/04}

\acmConference[HoTSoS '19]{Hot Topics in the Science of Security: Symposium and Bootcamp}{April 2--3, 2019}{Nashville, TN, USA}
\acmYear{2019}
\copyrightyear{2019}
\acmBooktitle{HoTSoS '19: Hot Topics in the Science of Security: Symposium
and Bootcamp, April 2--3, 2019, Nashville, TN, USA}

\acmArticle{4}
\acmPrice{}

\begin{document}
\title{Limitations on Observability of Effects\\in Cyber-Physical Systems}
\titlenote{Approved for Public Release; Distribution Unlimited. Public Release Case Number 18-4540}

\author{Suresh K. Damodaran}
\orcid{1234-5678-9012}
\affiliation{%
  \institution{MITRE}
  \streetaddress{MS M339, 202 Burlington Rd.}
  \city{Bedford}
  \state{MA}
  \postcode{01730}
}
\email{sdamodaran@mitre.org}
\author{Paul D. Rowe}
\orcid{1234-5678-9012}
\affiliation{%
  \institution{MITRE}
  \streetaddress{MS M339, 202 Burlington Rd.}
  \city{Bedford}
  \state{MA}
  \postcode{01730}
}
\email{prowe@mitre.org}

\renewcommand{\shortauthors}{S. K. Damodaran and P. D. Rowe}

\begin{abstract}
Increased interconnectivity of Cyber-Physical Systems, by design or otherwise, increases the cyber attack surface and attack vectors. Observing the effects of these attacks is helpful in detecting them.
In this paper, we show that many attacks on such systems result in a control loop effect we term Process Model Inconsistency (PMI). Our formal approach elucidates  the relationships among incompleteness, incorrectness, safety, and inconsistency of process models. We show that incomplete process models lead to inconsistency. Surprisingly, inconsistency may arise even in complete and correct models. 
We illustrate our approach through an Automated Teller Machine (ATM) example, and describe the practical implications of the theoretical results.
\end{abstract}

%
%



\settopmatter{printacmref=false}

\maketitle

\section{Introduction}  \label{sec:introduction}
Cyber-Physical Systems can range from industrial control systems (ICS) to Internet of Things (IoT) systems, and encompass a wide variety of  protocols, buses, and networks. While the definition of a Cyber-Physical System (CPS) is still evolving, we assume a CPS consists of interacting networks of physical devices and computational components that may be remotely controlled \cite{lee2008cyber}. While an earlier CPS may have been designed as a stand-alone and isolated system, modern CPS are designed with connectivity assumptions.  For example, the three-tier architecture for modern IoT systems described in \cite{lin2015industrial}  makes connectivity assumptions explicit. Due to intentional or unintentional network connectivity through the Internet or other means, no CPS can be assumed to be isolated.  The consequence of the increased interconnectivity among the systems is the addition of new cyber attack surfaces, and  vulnerabilities exploitable with new or existing attack vectors. 

To detect an attack, or evaluate the effect of an attack on a system, accurately observing the system state is very useful. A key question in this context is, "are there limitations on the observability of the system state that reduce the ability to conclude that an attack is happening or has happened?" Our main contribution in this paper is an answer to this question. We define Process Model Inconsistency (PMI) effect, and establish how PMI can be manifested as a result of several attacks identified in the literature (Section \ref{sec:process-model}).  In Section \ref{sec:states} we prove the limitations on observability of PMI effects. We illustrate how PMI effects may be masked with an example in Section \ref{sec:case-study}. In Section \ref{sec:implications}, we describe practical implications of limitations on observability, and we conclude in Section \ref{sec:conclusions}.

\section{Related Works}  \label{sec:related-works}

An effect is a consequence of an attack. Distinct attacks may result
in the same effect. This implies there could be physical effects as a
result of cyber attacks, and there could also be cyber effects on
system components as a result of physical attacks.  For example, a
tampered tire pressure gauge may show low tire pressure indicator on
the dashboard of a car, while there is perfectly adequate tire
pressure.  Cyber effects can manifest in various domains, ranging from
political instability to accidents. Ormrod et
al. \cite{ormrod2015system} present a System of Systems (SoS) cyber
effects ontology that attempts to capture the breadth of cyber effects
across physical, virtual, conceptual, and event domains within the
context of a battle. The cyber effects on human decision making,
originating from passive and active cyber attacks, is studied by
Cayirci et al. \cite{cayirci2011modeling}. Huang et al. have
analytically assessed the physical and economical consequences of
cyber attacks \cite{huang2009understanding}.

The physical effects on a CPS vary widely depending on the system and where the CPS is used, and therefore are difficult to categorize. In contrast, the cyber effects are relatively easier to categorize. The following categories of cyber effects are identified for Army combat training: Denial of Service (DoS), Information Interception, Information Forgery, and Information Delay \cite{marshallcobwebs2015}.  The cyber effects on the controller of the physical process change the operations of the process control system and the physical process in subtle ways. For process control systems, the security and protection of information is not enough, and it is necessary to see how the attacks affect estimation and control algorithms of a CPS, thus directly changing the physical world \cite{cardenas2011attacks}.  Han et al. \cite{han2014intrusion}  describe more obvious effects of cyber attacks such as draining out limited power of sensors, disrupted or incorrect routing, desynchronization, and privacy invasion through eavesdropping. Wardell et al. \cite{wardell2016method}   add changing of set points, sending harmful control signals, changing operator display to the list of effects. Cardenas et al. \cite{cardenas2008secure} describe physical attacks on sensors, actuators, or on the physical plant,  deception attacks carried through the compromises of sensors and actuators, and DoS attacks that make signals unavailable to the controller or physical process.  The effects of these attacks to the system could be missing or altered signals.  Mitchel and Chen  \cite{mitchell2016modeling} specify three types of failures in the context of modern power grids: \textit{attrition failure} when there are insufficient actuators or control nodes to apply control, \textit{pervasion failure} when the failed actuators and control nodes collude, and \textit{exfiltration failure} when the adversary obtains the grid data. Some of these cyber effects could also deceive a human operator who is part of the decision making process of the controller \cite{wardell2016method}. 

Developing a behavior model for a non-trivial CPS is a challenging problem because of the diversity of the system components, the phases of operations, and the need to reconcile the control, physical, software, and hardware models. Rajhans et al. \cite{rajhans2014supporting} provide a framework to integrate heterogeneous aspects of a system into a consistent verifiable behavior model. This problem is worsened in practice  due to the unavailability of any documented specifications for some system components. Pajic et al. \cite{pajic2017attack} discuss resilient statistical state estimation techniques while under attack as part of the DARPA HACMS project. Ability to sufficiently observe the properties of the CPS in operation is a requirement for accurate state estimation.

CPS effects have been analyzed by modeling a CPS as a multi-layered system with physical, sensor/actuator, network, and control layers  \cite{han2014intrusion,ashibani2017cyber}. Alternately, CPS effects have been  also studied by modeling it as a control system  \cite{cardenas2008secure,hahn2015multi,nourian2015systems}.  The control system based modeling approach  opens up specific ways to identify cyber effects and link them to attacks, as discussed in the next section. 

\begin{figure}

  \includegraphics[height=3in, width=3in]{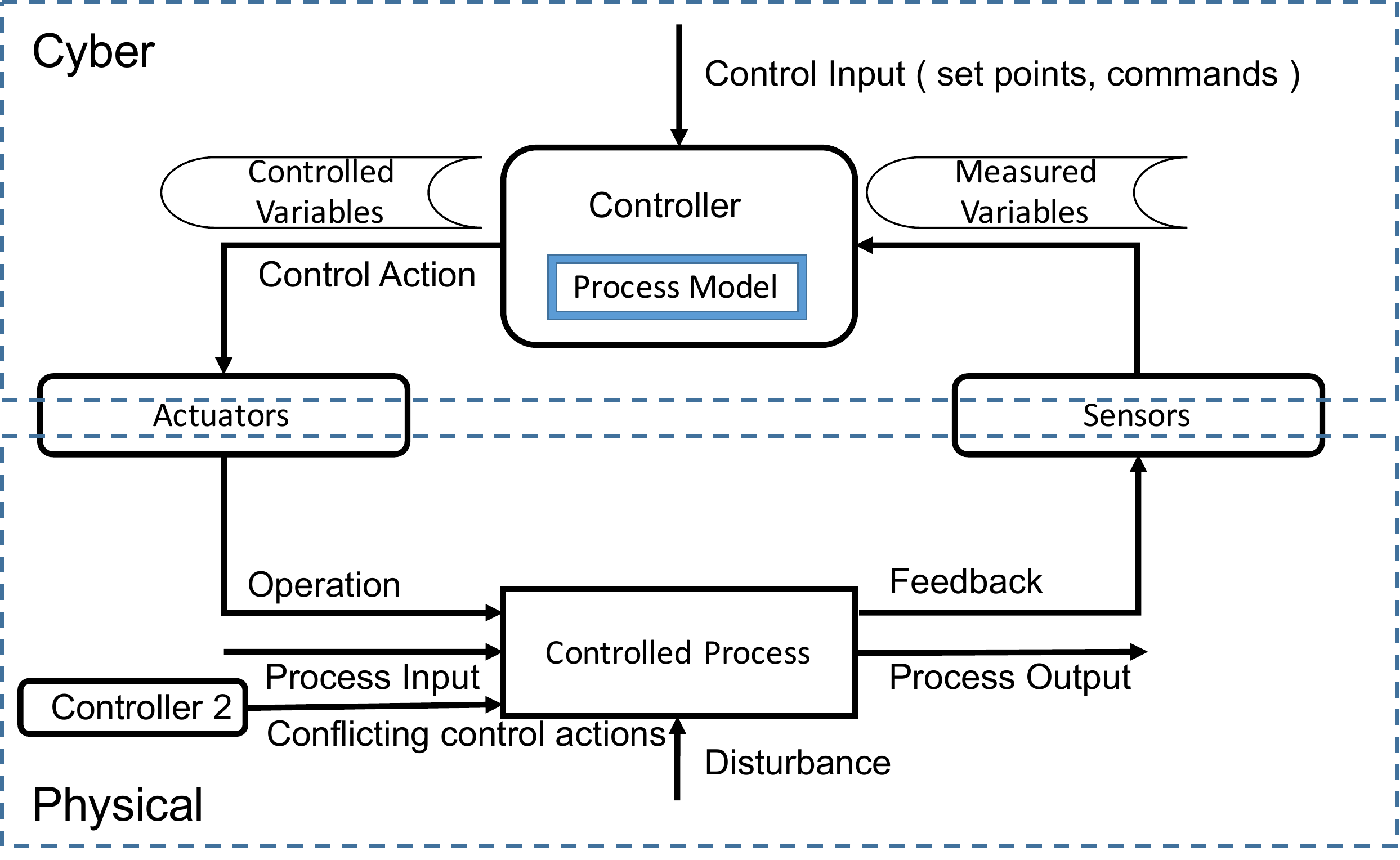}
  \caption{CPS Control Loop}  \label{fig:control-loop}

\end{figure}


\section{Process Model Inconsistency} \label{sec:process-model} Clark
and Wilson have argued \cite{clark1987comparison} that
\textit{external consistency}, the correspondence between the data
object and the real object it represents, is an important control to
prevent fraud and error. An attack may cause a change in the system
that we term an \textit{effect}. An effect may cause other changes in
the system that we term \textit{derived effects}. In this section, we
show why process model inconsistency is an important derived effect on
CPS through a control loop based analysis of cyber and physical
effects. In Section \ref{sec:control-loop}, we describe the control
loop view of CPS. We map the cyber and physical effects to derived
effects to control loop elements, and show in Section
\ref{sec:loop-impacts} that inconsistency of process model is an
important category of effect to a control loop.

\subsection{Control Loop View of  CPS}  \label{sec:control-loop}
A control loop, as shown in Figure \ref{fig:control-loop}, is a significant and discernible feature of any non-trivial CPS. A CPS that is a SoS may contain multiple such control loops that span multiple subsystems \cite{krotofil2013resilience}.  Figure \ref{fig:control-loop} describes  a simplified control loop derived from the descriptions of a control loop used for safety analysis \cite{leveson2011engineering}  and security analysis~\cite{nourian2015systems}. 

The \textit{controller} in Figure \ref{fig:control-loop} can include humans, or can be a fully-automated system, or it can be a semi-automated system with humans-in-the-loop. Stated differently, a human can be considered a part of the controller. Humans may also be participants in the \textit{controlled process}, responsible for providing sensory inputs to the \textit{controller}. The \textit{controller} may include estimation algorithms and controlling algorithms \cite{cardenas2008secure}. The \textit{controller} receives inputs from external entities such as set points, or other commands such as resets. All controllers must maintain a model of the \textit{controlled process}, called a \textit{process model}, within \cite{leveson2011engineering}. The dynamic behavioral aspects of the model are constructed and maintained by the \textit{controller} by processing the \textit{measured variable}s provided by the \textit{sensors}. The \textit{sensors} may provide these inputs to the \textit{controller} directly or indirectly through stored media such as logs. A \textit{controller} can represent the dynamic behavior of the \textit{controlled process}  in a \textit{process model}. A \textit{process model} is a behavioral model of the controller, and one way to represent it is by using a state diagram. DEVS formalism \cite{Zeigler:2000:TMS:580780} and  hybrid automata \cite{henzinger2000theory} use state diagrams to represent \textit{process model}s. 

 A state in a state diagram at time $t_0$ is identified by the values assigned to a set of state variables at $t_0$ by the \textit{controller}. These state variables are distinct from the \textit{measured variable}s in that the \textit{measured variable}s are communicated through the communication link between the \textit{sensors} and the \textit{controller}, while the state variables are maintained by the \textit{controller} within the \textit{process model}. The state variables have the important property that the values of the state variables at time $t_0$, combined with the values of \textit{measured variable}s collected during the time period between $t_0$ and time $t > t_0$ are sufficient to predict the values of state variables at time $t$, assuming the \textit{controller} processes the \textit{measured variable}s instantaneously.  The \textit{controller} provides the \textit{controlled variable}s to the \textit{actuators} which apply corresponding \textit{operation}s to the \textit{controlled process}. The system could include humans providing inputs to actuators as well. For example, in an elevator system, the users of the system can be considered to provide sensory inputs or \textit{measured variable}s to the controller by pressing buttons.

We highlight the upper half of Figure \ref{fig:control-loop} as  the \textit{cyber domain}, and the lower half as the \textit{\textit{physical domain}}, to signify the digital processing of information in the upper half.  The \textit{controller}, in this figure, is a digital information system that receives and produces digital values, while the \textit{controlled process}  is operated with analog inputs and processes. In a large system, however, designating a system component as either a cyber or physical component requires additional considerations.  For example, it is possible that the \textit{controlled process}  has some digital components, and the operational input from the \textit{actuators}  are digital values. In such cases, if the cyber elements are fully embedded in the physical system and  the cyber elements  are not explicitly modeled in the process model, we assume it is a physical component for the purpose of our analysis.

\subsection{Mapping Effects to Control Loop Effects}  \label{sec:loop-impacts}
Leveson \cite{leveson2013stpa} describes a number of things that can go wrong in the control loop from a safety perspective. While the  cyber attacks can cause the same kinds of effects as natural faults occurring with aging or fatigued components, analyzing cyber attacks can be more complex due to multiple reasons. Cyber attacks are intentional activities, as opposed to natural faults. Cyber attacks may cause multiple effects in different times  from the same attack, wheres the probability of occurrence of multiple natural faults is lower than the probability of occurrence of a single fault. Often the effects may be chained together to form a kill-chain. A kill-chain for a CPS may begin with a reconnaissance phase for collecting information flowing through the system as preparation for mounting an attack on the \textit{controlled process} \cite{hahn2015multi}. The effects caused by cyber attacks may seem unrelated as well, since the attacker has control over what effects are applied, and where. 

 In Table \ref{table:control-effects}, we map the different types of cyber effects  \cite{cardenas2008secure,han2014intrusion,marshallcobwebs2015,wardell2016method}, and their impacts to control loop components (Figure \ref{fig:control-loop}).  Cyber effects in CPS can also be caused by physical attack. For example, sensors or actuators may be damaged, tampered with, or attacked with electromagnetic pulse (EMP) \cite{cayirci2011modeling}, resulting in altering the \textit{measured variable} values or applying wrong or delayed operations to the \textit{controlled process}. Therefore, cyber effects in CPS can result from cyber, or physical attacks, as shown in Table \ref{table:physical-effects}.

\small
\begin{table}
  \caption{Control loop effects from cyber effects }
  \label{table:control-effects}
  \begin {tabular}  { | M{2cm} | M{2.4cm} | m{2.9cm}  |} 

    \hline
    \centering\textbf{Cyber Effects Type}
    & \centering\textbf{Control Loop Element}
    & {\centering\arraybackslash\textbf{Control Loop Effects}}  \\ \hline
    \multirow{3}{2cm}[-1.5mm]{\centering Information deception including forgery, spoofing, replay}
    & \centering Measured variable
    &  Inaccurate estimated \newline values and inconsistent \textit{process model}, Operator deception\\ \cline{2-3}
    & \centering Controlled variable
    & Altered controlled process    \\ \cline{2-3}
    & \centering Control input
    &  Altered controller\newline algorithm \\ \hline
    \multirow{2}{2cm}[-2mm]{\centering Information interception}
    & \centering Measured variable, Feedback
    & Inaccurate estimated values and inconsistent \textit{process model}, Operator deception\\ \cline{2-3}
    & \centering Controlled variable, Operation
    & Altered operation of the \textit{controlled process}   \\ \hline
    \multirow{2}{2cm}[-2mm]{\centering Information flooding (DoS)}
    & \centering Measured variable, Feedback
    &  Inaccurate estimated values and inconsistent \textit{process model}, Operator deception \\ \cline{2-3}
    & \centering Controlled variable, Operation
    & Altered operation of the \textit{controlled process}  \\ \hline
    \multirow{2}{2cm}[2mm]{\centering Information timing including delay, desynchronization}
    & \centering Measured variable
    &  Inaccurate, delayed estimated values and inconsistent \textit{process model}, Operator deception \\ \cline{2-3}
    & \centering Controlled variable, Operation
    &  Altered operation of the \textit{controlled process}  \\ \hline
    \centering Information exfiltration
    & \centering Measured variable, Feedback,
      Controlled variable, Operation
    & Privacy violation \\ \hline
    \centering Process input tampering
    & \centering Controlled process
    & Altered controlled process \\ \hline
    \centering Process output tampering
    & \centering Controlled process
    & Altered \textit{process output} \\ \hline
    \centering Information interception
    &\centering Control input
    & Conflicting control inputs and incorrect control behavior   \\ \hline
  \end{tabular}
\end{table}

\normalsize

\small
\begin{table}[h]
    \caption{Control loop effects from physical effects }
      \label{table:physical-effects}
  \begin {tabular}  { | M{1.8cm} | M{1.8cm} | m{3.1cm}  |} 
    \hline
    \centering\textbf{Physical Effects Type}
    & \centering\textbf{Control Loop Element}
    & \centering\arraybackslash\textbf{Immediate Cyber \& Control Loop Effects}  \\ \hline
    \multirow{2}{1.8cm}[-4.5mm]{\centering Physical tampering}
    & Sensor
    & Information deception,\newline Information interception, Information
      timing\\ \cline{2-3}
    & Actuator
    & Information interception, Altered operation of the
      \textit{controlled process}, Attrition failure, Pervasion
      failure \\ \hline
    \multirow{2}{1.8cm}[-4.5mm]{\centering Drained power}
    & Sensor
    &  Information deception,\newline Information interception, Information
      timing\\ \cline{2-3}
    & Actuator
    &   Information interception, Altered operation of the \textit{controlled process}, Attrition failure, Pervasion failure \\ \hline
  \end{tabular}

\end{table}
\normalsize

One common category of effect to the control loop is the inconsistent
\textit{process model}, as seen from Table
\ref{table:control-effects}. Let us explore this effect further. Young
and Leveson \cite{young2014integrated} note that many accidents stem
from the inconsistencies between the \textit{process model} and the
state of the \textit{controlled process}.  Inconsistency of the
\textit{process model} can represent effects stemming from sensor
tampering, or effects on the \textit{measured variables} or the
feedback to the sensors, as shown in the cyber effects in Table
\ref{table:physical-effects}, and the impact of these cyber effects in
Table \ref{table:control-effects}. Changes to the actuators and the
\textit{controlled variable}s can also result in altered
\textit{controlled process} behavior. These effects may change the
\textit{process model} to be inconsistent as well, since the
operations that were applied to the \textit{controlled process} could
be different from what the commands to actuators suggested.  However,
not all effects could be represented by inconsistency of the
\textit{process model}. In particular, the tampering of
\textit{process input} directly into the \textit{controlled process}
may not lead to an inconsistent \textit{process model} if a sensor is
able to pick-up the changes in the \textit{controlled process} that
result from this input.  Leveson \cite{leveson2013stpa} points out
that a \textit{process model} can be incomplete or incorrect. If the
\textit{process model} is incomplete or incorrect, it can be also
inconsistent with the state of the \textit{controlled process} even
without cyber effects applied. Even if the \textit{process model} is
(statically) complete and correct, inconsistencies may arise
dynamically during its operation. This may be due, for example, to
adversarial modifications of signals causing a controller to have a
false view of the state transition actually taken by the controlled
process, even if the actual transition taken is possible in the
\textit{process model}. In the next section, we formalize the concept
of inconsistent process model as Process Model Inconsistency (PMI)
effect, and prove its properties.

\section{Process Model, States, and Properties}   \label{sec:states}
%
In the previous section, we identified inconsistent process model as
an important control loop effect caused by other effects. In this
section we define this effect precisely so that the limitations on the
observability of this effect can be studied.  As we described in
Section \ref{sec:control-loop}, a cyber-physical system has at least
one control loop, with messages or signals transmitted among the
components participating in the control loop. These components may be
distributed geographically, or logically separated, and hence we need
to model the dynamic aspects of a cyber-physical system as a
distributed system of communicating components, where each such
component can be a system with its own subcomponents. A
\textit{process model} describes this dynamic behavior of the system,
and a state space may be used to represent the \textit{process model}
and its properties.  A process model may be used to represent both
(a)~the controller's belief about the state of the controlled process,
(Figure~\ref{fig:control-loop}) and (b)~the actual state of that
process. In this section, we formalize process model inconsistency
(PMI) as a discrepancy between these two models and characterize the
ways in which it can arise.
\subsection{Process Model and
  Observability}\label{sec:process-model-obs}
The term \textit{process model} is used to describe the dynamic behavior
of an individual component of a system. We use the term \textit{global
  process model} for the synthesized process model of all the
components.  A formalism to study such resultant behavior of
synthesized process models in System of Systems (SoS) is the Discrete
EVent Simulation (DEVS) formalism, applicable to digital and analog systems \cite{Zeigler:2000:TMS:580780}. 

The DEVS formalism accomplishes SoS modeling by defining two types of
models: atomic and coupled. When a system cannot be decomposed any
further, its behavior is specified through an atomic DEVS model. A
coupled model allows SoS constructs to be built as a hierarchical
structure that comprises atomic and coupled models.  Let us review
atomic DEVS and coupled DEVS model definitions
\cite{Zeigler:2000:TMS:580780} below:
\begin{definition}
$DEVS_{atomic}= <X,S,Y,\lambda, \delta_{int},\delta_{ext}, \delta_{con}, ta>$, where 
\begin {itemize} 
\item $X$ is the set of inputs described in terms of pairs of port and value: $\{p,v\}$, 
\item $Y$ is the set of outputs, also described in terms of port and value: $\{p,v\}$, 
\item $S$ is the state space that includes the current state of the atomic model, 
\item $\delta_{int}: S\rightarrow S$ is the internal transition function,
\item $\delta_{ext}: Q \times X^b \rightarrow S$ is the external
  transition function that is executed when an external event arrives
  at one of the ports, changing the current state if needed,
  $Q = \{ (s,e)| s \in S$,$0 \le e \le ta(s)\}$ as the total state
  set, where $e$ is the time elapsed since the last external
  transition, and $X^b$ is the set of bags over elements in $X$;
\item $\delta_{con}: S \times X^b \rightarrow S$ is the confluent
  function, subject to
  $\delta_{con} (s,\varnothing) = \delta_{int}(s)$ that is executed if
  $\delta_{ext}$ and $\delta_{int}$ end up in collision; and
\item $\lambda: S \rightarrow Y$ is the output function that is
  executed after internal transition function is completed,
\item $ta(s): {R}_{0,\infty}^+ $ is the time advance function.
\end{itemize}
    \label{def:patomic}
\end{definition} 

\begin{definition} \label{def:pcoupled}
  $DEVS_{coupled} = <X,Y,D,\{M_{d}\},\{I_{d}\},\{Z_{i,d}\}> $, for
  each component model, $d \in D$, where,
\begin{itemize}
\item $D$ is a set of labels assigned uniquely to each component model
\item $M_d$ is a DEVS model of $d$
\item $I_d$ is the \textit{influencer} set of $d: I_d \subseteq D \cup \{DEVS_{coupled}\}, d \notin I_d$, and for each $i \in I_d$,\\
  $Z_{i,d}$ is an $i$-to-$d$ coupling such that\\
  \hspace*{1cm}$Z_{i,d}: X \rightarrow X_d$ is an external input coupling (EIC), if $i = DEVS_{coupled}$\\
  \hspace*{1cm}$Z_{i,d}: Y_i \rightarrow Y$ is an external output coupling (EOC), if $d= DEVS_{coupled}$\\
  \hspace*{1cm}$Z_{i,d}: Y_i \rightarrow X_d$ is an internal coupling
  (IC), if $i \ne DEVS_{coupled}$ and $d \ne DEVS_{coupled}$
\end{itemize}
\end{definition}

Hybrid automata theory \cite{henzinger2000theory} is another
established formalism used in robotics for behavior modeling. In both
formalisms, the dynamic behavior of the system is defined using
states.  Therefore, we use states for representing process models.

In a distributed system, the inputs are delivered to system components
using messages, and output messages are generated by system
components. Therefore, in Definition \ref{def:patomic}, the inputs and
outputs are messages. We assume that a message contains values
assigned to a set of message variables. The components of a system
also have component variables whose values determine the state of that
component. The values in message variables are mapped to and from
component variables as the system runs. In CPS, the component
variable values of the controller are either used to store (1) the values of \textit{measured variables}  encapsulated in messages from sensors, or (2) the commands or estimations of \textit{controlled variables} (Figure
\ref{fig:control-loop}). We discuss the concept of state in more detail below.

%

In Definition \ref{def:patomic}, $S$ is the state space of an atomic
component of a cyber-physical system. As a running example, let us
consider an elevator with a simple control loop. The controller must
maintain a process model for the various components it controls. Among
those is the elevator car. We assume this process model encodes two
possible states of the elevator car: RUNNING or STOPPED. Its internal
model of the full system would also contain states for other
components such as the status of the doors on various floors. The
global state would then be an $n$-tuple of local states
$\Sigma = (S_1,..., S_n)$, where $S_i$ is the local state of the
component $i$~\cite{cooper1991consistent,babaoglu1993consistent}. For
simplicity, we focus only on models for atomic components since PMI
can already arise in this case. We denote the state space of an atomic
model as~$S$. 


The controller determines this local state from the values stored in a
sequence of component variables, wherein each constituent component variable
has a set of potential values. A component variable may be assigned values
during the system operation, which may be normal operation or abnormal
operation caused by natural failures or cyber attacks. These values
may be discrete, or analog. For analog values such as temperature
readings, we assume they are discretized. In our case, we imagine the
controller populates two variables (statusCarDoor, and motorRunning)
based on data from sensors. The cartesian product of all of the
possible values of these component variables for each component defines
all the possible combination of values these variables can assume. We
define a multivariable to capture this concept.
\begin{definition}
\label{def:multivariable-set}
A \textit{multivariable}, ${V}^{S}$ = $<v^1,v^2,...,v^m>$ may
be defined for a state space $S$, where $v^j$ may assume any
values from the set of values in $P^j$. The potential ordered
m-tuple values the multivariable can assume are in the multivariable
space $P$ = ($P^1 \times P^2 \times ... \times P^m$). Elements $p\in
P$ are called \emph{variable assignments}.
\end{definition}
In the elevator example, the controller maintains two component variables
to store its state: $<$statusCarDoor, motorRunning$>$.  The
statusCarDoor variable can assume any value from the set
$P^1=$ \{CLOSED, OPEN\}, and the motorRunning variable can assume any
value from the set $P^2=$ \{ON, OFF\} (Table~\ref{tab:carctrl}).
In the table, for brevity, we use "*" to denote all possible values
of that variable.

The mapping of variable assignments to the states of the elevator car
are also shown in Table~\ref{tab:carctrl}. This mapping is the state model, defined below. The constituent variables of a multivariable used in such mapping in the state model are also referred to as \textit{state variables}.
\begin{table}
  \caption{Elevator Car Controller Multivariable}
  \label{tab:carctrl}
  \begin{tabular}{ccl}
    \toprule
    statusCarDoor&motorRunning&state\\
    \midrule
    CLOSED & * & RUNNING\\
    OPEN & OFF& STOPPED\\
  \bottomrule
  \end{tabular}
\end{table}
\begin{definition}
\label{def:state-model}
A \emph{state model} is a triple ($P$, $F$, $S$) where $F: P
\rightarrow S$ a surjective partial observation function.
A variable assignment $p\in P$ is \emph{observable} iff $p$ is in the
domain of~$F$. 
\end{definition}
The use of a partial function~$F$ in
Definition~\ref{def:state-model} is important. Variable assignments
not in the domain of~$F$ typically correspond to combinations of
values assigned to variables that are not thought to be possible. For
example, in the model of the elevator car maintained by the car
controller, there is no state associated to the variable assignment
(OPEN, ON) because the car motor should never be running while the
door is open. An implementation may have enough foresight to include
an ERROR state in the state space~$S$. This is easy enough to do with
simple models. However, with more complex models, some variable
assignments may not be observable because the programmers did not
think they were possible. Exactly how the state gets updated (or not)
will depend on the details of the implementation. For the purposes of
our formalization, it is sufficient to allow the observation function
to be partial and consider variable assignments outside the domain
of~$F$ to be unobservable.
\begin{definition}
  \label{def:known-process-model}
  A \emph{process model} $(P,F,S,T)$ combines a state model with a
  transition relation $T\subseteq S\times S$. 
\end{definition}
\subsection{Incorrectness and
  Incompleteness}\label{sec:incompleteness}
The process model maintained by a controller represents only those
states and transitions the controller knows to expect. We therefore
refer to the process model maintained by a controller as a
\emph{known} process model ($P_k, F_k, S_k, T_k$). Unfortunately,
reality is almost always richer than this model. For example, the
``true state space'' of an elevator car would be more than just the
set \{RUNNING, STOPPED\}. It would capture whether the car was moving
UP or DOWN, what floor it is at, if it is between floors, etc.

Considering all possible measurable variables of a system we may
imagine a \emph{potential state space} $S_P$ resulting from a state
assignment function $F_P$ that serves as an upper bound on what states
the controlled component can actually be in. In general this potential
state space may not even be finite. For example, using the natural
numbers to represent the possible values for the floor representing
the elevator car's position while stopped gives an infinite set of
possibilities. In practice, reality is bounded in various ways. For
example, in a 5-story building the elevator car can never be on
floor~6. When defining the ``true'' state of a controlled component,
we may imagine that there is some finite set of variables that can be
measured, and that access to those values would provide an accurate
picture of reality, that we call the \textit{ground truth process model}.

\begin{definition}\label{def:ground-truth-model}
  The \emph{ground truth process model} for a component is a
  process model $(P_r, F_r, S_r, T_r)$ such that every state of $S_r$
  is reachable via $T_r$ from some initial state $s_0$. That is, for
  every state $s_n\in S_r$, there is some sequence
  $s_0,s_1,\dots,s_{n-1},s_n$ such that $(s_i,s_{i+1})\in T_r$ for
  every $0\leq i\leq n$. 

\end{definition}

The ground truth process model is not typically known a priori. For a
non-trivial system it may be hard to create an exhautive ground truth
process model consisting of all reachable states. We would like to
stress that the \textit{ground truth process model} is not expected to
be built a priori, or at any time, by the system designers or
operators. Rather, the ground truth process model is the hypothetical
and accurate behavior model of the system under normal and some
abnormal operations. The purpose of defining the ground truth process
model is to contrast it with the \textit{known process} model. There
can also be multiple \textit{ground truth} process models for the same
system under different abnormal operations (e.g. under different
adversarial assumptions). When new abnormal operations happen,
naturally the \textit{ground truth} process model will also expand to
include the new reality. Therefore, one may think of the
\textit{ground truth} process model as the process model of the system
under the normal and abnormal operations we are considering for
analysis.

A known process model and a ground truth process model can differ in
two primary ways. The known process model can either be incorrect,
incomplete, or both, as defined below.

\begin{definition}\label{def:incorrect-incomplete}
  Let $(P_k, F_k, S_k, T_k)$ and ($P_r, F_r, S_r, T_r$) be known and
  ground truth process models respectively. The known process model is
  \emph{incorrect} iff $S_k\setminus S_r \ne \emptyset$ or
  $T_k\setminus T_r \ne \emptyset$. The known process model is
  \emph{incomplete} iff $S_r\setminus S_k \ne \emptyset$ or
  $T_r\setminus T_k \ne \emptyset$. When the known process model is
  incomplete, a \emph{forced state} is any state
  $s\in S_r\setminus S_k$ and a \emph{forced transition} is any
  transition $t=(s_1,s_2)\in T_r\setminus T_k$.
\end{definition}

Intuitively, incorrectness refers to errors, and incompleteness refers to missing transitions and states in the known process model. In practice, it is possible for the known process model maintained by
the controller to be incorrect in the sense of
Definition~\ref{def:incorrect-incomplete}. However, in the context of
security incorrectness is not always a cause for concern. If one can make security
guarantees based on assuming the system can reach more states than it
really can, this typically means that those security guarantees would
still hold in the more restricted system without those states.

Incompleteness is somewhat different from incorrectness. It is a
simple application of the definitions to note that for an incomplete
known process model, there must be either a forced state or a forced
transition. These are depicted in Fig.~\ref{fig:forced-transition},
where known states and transitions are represented with solid lines,
and forced states and transitions are represented with dotted
lines. Thus, $s_r$ is forced state, and $(s_1,s_r)$ is forced
transition into the forced state from known state $s_1$. Similarly,
$(s_3, s_2)$ is a forced transition between known states.

\begin{figure}

    \includegraphics[height=1.5in, width=1.5in]{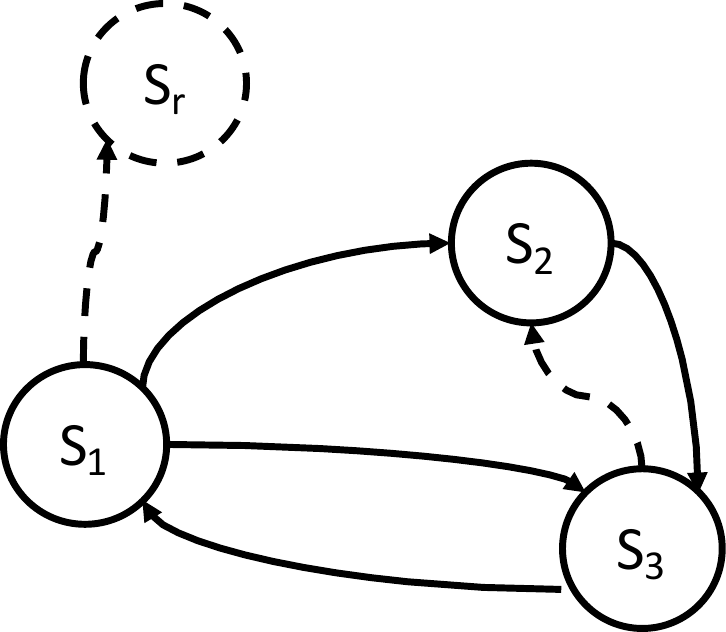}

  \caption{Forced State and Transition}
    \label{fig:forced-transition}
\end{figure}

Forced states and transitions naturally pose potential dangers that
incorrectness doesn't. Namely, safety properties satisfied by the
known process model may not be satisfied by the ground truth process
model. In a distributed system, safety
properties~\cite{alpern1987recognizing} of the system must evaluate to
true in the states of the system during operations. Safety properties
are inherent properties of any given state, and are not dependent on
the details of attacks or failures that caused the system to enter
that state. We apply this concept of safety properties to the states
of the controlled process. Hybrid automata use the concept of a
\textit{bad state} to describe states where at least one safety
property is violated~\cite{raskin2005introduction}.  We define
\textit{bad state} and \textit{normal state} formally below.
\begin{definition}
  \label{def:normal-state}
  A \textit{normal state}, $s_{n} \in S_{n} \subseteq S_P$, is a state
  where all safety properties will evaluate to true. $S_n$ is the
  normal state space. A \textit{bad state},
  $s_{b} \in S_{b} \subseteq S_P$, is a state where at least one safety
  property will evaluate to false. $S_b$ is the bad state space.
\end{definition}

The definition of the bad states makes no reference to the known or
ground truth process models. The general case (assuming a correct
known process model) is depicted in
Fig.~\ref{fig:state-space}. Notice, in particular, that there may be
many states in the ground truth process model that are unknown, yet
are not bad states. That is, forced states are not necessarily bad
states. This means that the mere fact of entering unknown states need
not violate safety properties of interest. However, if
$(S_r\setminus S_k) \cap S_b \neq \varnothing$, the attacker can
succeed.

For the elevator example, consider the situation depicted in
Fig.~\ref{fig:CarCtr-process-model}. We could define a safety
property: "Elevator Car will not move with the car door open." Given
this safety property, the states with solid outlines (described in
Table~\ref{tab:carctrl}) are normal states. A Car Controller state X
corresponding to $<$OPEN,ON$>$ is a \textit{bad} state, since we don't
want the car running with its door open. In this example, the bad
state X in Fig.~\ref{fig:CarCtr-process-model} is a forced state
because there is some way to transition into it from the STOPPED
state. An adversary able to force a transition into this state would
violate the safety property. Alternate safety property definitions may
result in differing \textit{normal} and \textit{bad} state
designations.

\begin{figure}

    \includegraphics[height=1.5in, width=1.5in]{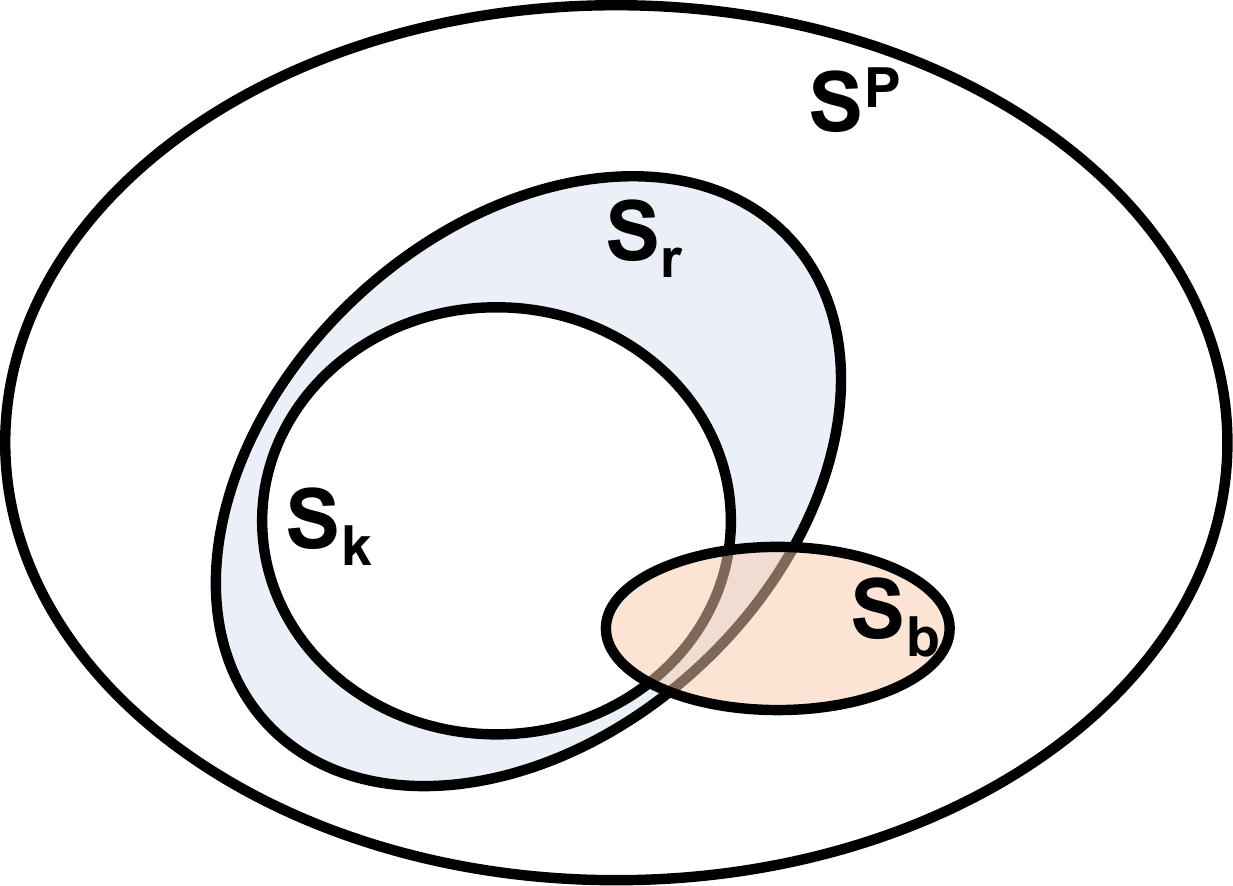}

  \caption{State Space}
    \label{fig:state-space}
\end{figure}

The constituent variables of the multivariable may be analyzed
directly for the evaluation of safety properties. That is, rather than
evaluate a safety property against a state, it may be evaluated
against a \emph{variable assignment}. Regardless of whether a variable
assignment is in the domain of $F_r$ or $F_k$, safety properties may
be evaluated against them. It should be noted, however, that for
consistency in this case, when two variable assignments get mapped to
the same state, they should either both satisfy the safety property or
both fail the safety property.

\begin{figure}

    \includegraphics[height=1in, width=1in]{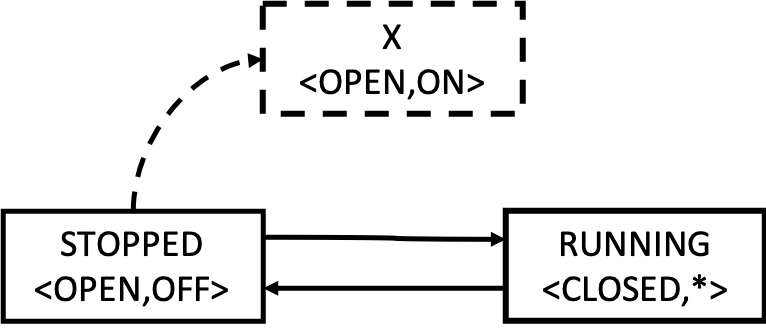}

  \caption{Car Controller State Space}
    \label{fig:CarCtr-process-model}
\end{figure}

\begin{lemma}
  \label{lemma:inconsistency}
  Let $(P_k, F_k, S_k, T_k)$ and $(P_r, F_r, S_r, T_r)$ be a known
  process model and a ground truth process model, respectively. If
  there is a forced state, then there is a forced transition.
\end{lemma}
\begin{proof}
  Let $s_n\in S_r$ be the forced state assumed to exist. By
  Definition~\ref{def:ground-truth-model}, there is some sequence
  $s_0,s_1,\dots,s_{n-1},s_n$ such that $(s_i,s_{i+1})\in T_r$ for
  every $0\leq i\leq n$. In particular $t=(s_{n-1},s_n)\in T_r$ but
  $t\not\in T_k$ because $T_k\subseteq S_k\times S_k$ and $t\not\in
  S_k\times S_k$.
\end{proof}
In the elevator example, the state X in the \textit{ground truth}
process model for Car Controller (Figure \ref{fig:CarCtr-process-model}),
cannot be disconnected from the other states, since this state is
reachable in the \textit{ground truth} process model due to a cyber attack or
a system failure. This transition is shown by the dotted arrow from
STOPPED state to X state.

\subsection{Process Model Inconsistency}\label{sec:ground-truth}

In the previous section, our treatment of incorrectness and
incompleteness of the known process model with respect to a ground
truth process model naturally invoked the notion of safety
property. Young and Leveson~\cite{young2014integrated} note that
problems arise due to an \emph{inconsistency} between the known
process model and the ground truth. This holds irrespective of whether
or not particular safety properties are violated. That is, process
model inconsistency (PMI) is a potentially deleterious effect in
itself. In this section we therefore formally define PMI and prove
that it necessarily poses a danger for all incomplete known process
models.

Intuitively, PMI occurs when the observations made in the known
process model differ from those of the ground truth process model. In
order to talk meaningfully about observations of the ground truth
\emph{from within} the known process model, we need a way of
connecting the two process models to define the known process model's
observations of the ground truth. Since observations in the known
process model correspond to interpretations of variable assignments
$p \in P_k$, it is sufficient to connect the ground truth variable
space $P_r$ with the known variable space $P_k$.

In general, these two spaces may not be related. For example, the
variables tracked in the known process model might not be direct
measurements. However, it is with no loss of generality that we may
assume the ground truth variable space to be a superset of the known
variable space. We can always expand $P_r$ to contain known variables
not otherwise present, and simply allow $F_r$ to be insensitive to the
values of these extra variables. This will not interfere with the
established aspects of the ground truth model.

We formally define this structured connection between known process
models and ground truth models below.

\begin{definition}\label{def:submodel}
  Let $(P_k, F_k, S_k, T_k)$ and ($P_r, F_r, S_r, T_r$) be known and
  ground truth process models respectively. The models are
  \emph{connected by inclusion and projection} (or just
  \emph{connected}) iff $P_k = P_k^1\times P_k^2 \cdots \times P_k^m$,
  and $P_r = P_r^1\times P_r^2 \times \cdots \times P_r^n$, where
  $n > m$ and $P_k^i\subseteq P_r^i$ for $1\leq i\leq
  m$. $\iota:P_k\hookrightarrow P_r$ is the natural inclusion of $P_k$
  into $P_r$, where fixed values for the variable $v^{m+1},\dots,v^n$
  are chosen. $\pi:P_r \twoheadrightarrow P_k$ is the inverse
  (partial) function.
\end{definition}

Since the order of presentation of the $P^i$ is arbitrary, we choose a
consistent order for $P_k^i$ and $P_r^i$ to ensure $\pi$ and $\iota$
work component-wise in the natural way.  When two models are
connected, their connection can be depicted as in
Fig.~\ref{fig:PMI}. The function $\iota$ and $\pi$ allow us to
``move'' from one model to the other. This ultimately allows a clean
definition of inconsistency. In particular, we can start with a
well-defined notion of observations of ground truth within the known
process model.

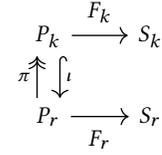
\begin{figure}
  \[
    \begin{tikzcd}
      P_k\ar[d,"\iota",hookrightarrow, shift left=2]\ar[r]\ar[r,phantom,shift left=4,"F_k"] & S_k\\
      P_r\ar[u,twoheadrightarrow,"\pi", shift left=2]\ar[r]\ar[r,phantom, shift right=4,"F_r"'] & S_r
    \end{tikzcd}
  \]
  \caption{Connected Process Models}
    \label{fig:PMI}
\end{figure}

\begin{definition}\label{def:observable}
  A variable assignment $p\in P_r$ is \emph{observable} in the known
  process model iff $p\in \mathop{dom}(\pi)$ and
  $\pi(p)\in \mathop{dom}(F_k)$. $p\in P_r$ is \emph{correctly
    observed} in the known process model iff $p$ is observable and
  $F_k(\pi(p)) = F_r(p)$. A pair of variable assignments ($p_a,p_b$)
  is observable iff its component variable assignments are
  observable. Similarly the pair $(p_a,p_b)$ is correctly observed iff
  its component variable assignments are correctly observed.
\end{definition}

This definition uses $\pi:P_r\twoheadrightarrow P_k$ to mediate
observations. Inconsistencies may arise because we cannot ensure that
$F_r$ and $F_k$ provide consistent interpretations of the common
portions of variable assignments. That is, observations in the known
process model based on access to variables in $P_k$ interpreted
according to $F_k$ might differ from observations in the ground truth
process model based on the larger set of information in $P_r$ with
accurate interpretation~$F_r$.

There are two main things that give rise to disagreements. The current
state may simply be unobservable from within the known process model,
or it may make an incorrect observation. 

Since $P^i_k \subseteq P^i_r$ for $1\leq i\leq m$, it is possible for
this function to be undefined for some values. Consider for example a
variable assignment $p = (p^1,p^2,\dots,p^n)\in P_r$ where
$p^1\in P_r\setminus P_k$. The resulting variable assignment
$(p^1,p^2,\dots,p^m)$ is not in $P_k$, so $p$ is not in the domain of
$\pi$. This results in the variable assignment not representing an
observable state in the known process model. Even if $\pi$ is defined
on a given $p\in P_r$, it is possible that $\pi(p)$ is outside the
domain of the partial function $F_k$. Here too, the variable
assignment does not represent an observable state. When the variable
assignment is observable, the known process model assigns it a
definite state. An incorrect observation is one in which this state
disagrees with the state observed by the ground truth model.

\begin{definition}
  \label{def:process-model-inconsistency}

  Let $(P_k, F_k, S_k, T_k)$ and $(P_r, F_r, S_r, T_r)$ be known and
  ground truth process models respectively that are connected. A
  transition $t=(s_a,s_b)\in T_r$ is an instance of \emph{Process
    Model Inconsistency} (PMI) iff there are some variable assignments
  $(p_a, p_b)$ resulting in $t$ and $(p_a,p_b)$ is unobservable or
  incorrectly observed in the known process model.
\end{definition}
Sometimes, cyber effects may need multiple \textit{forced states},
\textit{forced transitions}, or both to describe them, as we show in
the example in Section \ref{sec:case-study}. The theorem below
explores the limitations on observability of forced transitions in
\textit{ground truth} model given that we are only equipped with the
\textit{known process model}.

\begin{theorem}
  \label{theorem:detection}
  Let $(P_k, F_k, S_k, T_k)$ be an incomplete model with respect to
  the ground truth model $(P_r, F_r, S_r, T_r)$.  Assume they are
  connected by inclusion and projection. Then either the ground truth
  model contains at least one instance of PMI, or there is a forced
  transition that is correctly observed.

  \begin{proof}
  By Def.~\ref{def:incorrect-incomplete}, the ground truth process model
  contains either a forced state or a forced transition. But
  Lemma~\ref{lemma:inconsistency} tells us that the existence of a forced
  state implies a forced transition. So we know there is some
  $(s_a,s_b)\in T_r\setminus T_k$. We take cases on whether or not
  $s_a\in S_k$ and $s_b\in S_k$. We examine observability according to
  Def.~\ref{def:observable} using the map $\pi:P_r\twoheadrightarrow
  P_k$ which must exist according to Def.~\ref{def:submodel}.
  \\

  \textbf{Case 1.} At least one of $s_a$ or $s_b$ is in
  $S_r\setminus S_k$. Without loss of generality, let it be $s_a$. We
  will now establish what the known process model might
  observe. Consider any element $p\in P_r$ giving rise to $s_a$ (i.e
  $F_r(p)=s_a$). Such a $p$ exists because $F_r$ is surjective by Def. \ref{def:state-model}.

  \textbf{Case 1a.} If $p$ is not observable, then
  ($s_a,s_b$) is an instance of PMI by Def. \ref{def:process-model-inconsistency}.

  \textbf{Case 1b.} If $p$ is observable, then
  $F_k(\pi(p))=s_a'\in S_k$ is well-defined. But since
  $s_a\not \in S_k$, $s_a\ne s_a'$, the transition is incorrectly
  observed (Def. \ref{def:observable}). Thus $(s_a,s_b)$ is an instance of PMI by Def. \ref{def:process-model-inconsistency}.
  \\

  \textbf{Case 2.} Both $s_a$ and $s_b$ are in $S_k$. Consider any
  $(p_a,p_b)$ such that $F_r(p_a) = s_a$ and $F_r(p_b) = s_b$.

  \textbf{Case 2a.} At least one of $p_a$ or $p_b$ is unobservable. In
  this case, $(s_a, s_b)$ is an instance of PMI by Def. \ref{def:process-model-inconsistency}.

  \textbf{Case 2b.} Both $p_a$ and $p_b$ are observable. Thus we can
  define $s_a' = F_k(\pi(p_a))$ and $s_b' = F_k(\pi(p_b))$.

  \textbf{Case 2b(i).} Either $s_a' \ne s_a$ or $s_b' \ne s_b$. In
  this case, $(s_a,s_b)$ is incorrectly observed and, hence, it is an
  instance of PMI by Def. \ref{def:process-model-inconsistency}.

  \textbf{Case 2b(ii).} $s_a' = s_a$ and $s_b' = s_b$. In this case,
  the transition is correctly observed as $(s_a, s_b)$, but this
  transition was assumed to be a forced transition, so the last clause
  of the conclusion is satisfied.

\end{proof}
\end{theorem}

When a forced transition is observable, depending on the
implementation of the \textit{known process} model, this state
transition may be flagged as an error, or ignored.  Therefore, the
transition from the STOPPED state to X state in Car Controller is a
forced transition (Figure \ref{fig:CarCtr-process-model}). Given Theorem \ref{theorem:detection}, describing
instances of PMI using a state diagram of the \textit{known state}
space poses some interesting challenges, since the state diagram of
the \textit{known state} space can only be used to show forced
transitions among the known states.

Interestingly, the converse of Theorem~\ref{theorem:detection} is
false. That is, incompleteness is not necessary for PMI. Indeed PMI
can occur even when the known process model is both correct and
complete. This could be due to a tampered sensor sending a false
signal to the controller. Expanding on the elevator example, the
controller may receive a signal that the elevator car went 1~floor up,
when, in fact, it went 1~floor down. Both are possible, so there is no
inherent problem with the process model itself. Rather, PMI arises in
this instance due to differences in the observation functions $F_k$
and $F_r$. This means that improving the known process model may be
insufficient to fully address instances of PMI! It must be addressed
by fixing the system as a whole.



 \section{Illustrative Example}   \label{sec:case-study}
We will presently illustrate the theoretical results from the previous section using the example of an Automated Teller Machine (ATM) state machine. We use the ATM state machine described by Iqbal et al.~\cite{iqbal2017formal} for our illustration, since this model is formally verified. Figure \ref{fig:ATM-sm} reproduces the state machine from  \cite{iqbal2017formal} with minor modifications for readability.  Since this state machine is designed by the model developers of the ATM, therefore, this state machine  is the \textit{known process} model by Definition \ref{def:known-process-model}.

In this ATM model, the states of the ATM are specified in the rectangle boxes, and the state transitions are described by annotated arrows. The customer is not explicitly modeled, though the state diagram implies an external customer. The annotations describe either the conditions for a state transition, or the inputs to a state that could cause a state transition. These states define the state space, $S_k$.  Some of the arrows may be interpreted as internal transition functions, $\delta_{int}$. For example, ("Wrong PIN","Print Receipt") as in Figure \ref{fig:ATM-sm}. Some other transitions may be interpreted as external transfer functions $\delta_{ext}$. For example, when the "Insert Readable Card" external event occurs, the system transitions to "Request Password" state. There are $\lambda$ output functions that map the states to external outputs. For example, from "Verify Account" to an external system to "Verify Externally." 
In this example, we do not show the state variables, and assume a state assignment function exists for the \textit{known process} model. With respect to a safety property of "ATM will dispense cash to authenticated users," all the states in the diagram are normal states (Definition \ref{def:normal-state}). On the other hand, if the safety property is "Only authenticated users are allowed access to the ATM", the state "Wrong PIN" may be considered a \textit{bad} state. In the \textit{known process} model, this transition will be observed as a known transition from "Process Transaction" to "Dispense Cash."
 
\begin{figure}
  \includegraphics[height=3.5in, width=3.5in]{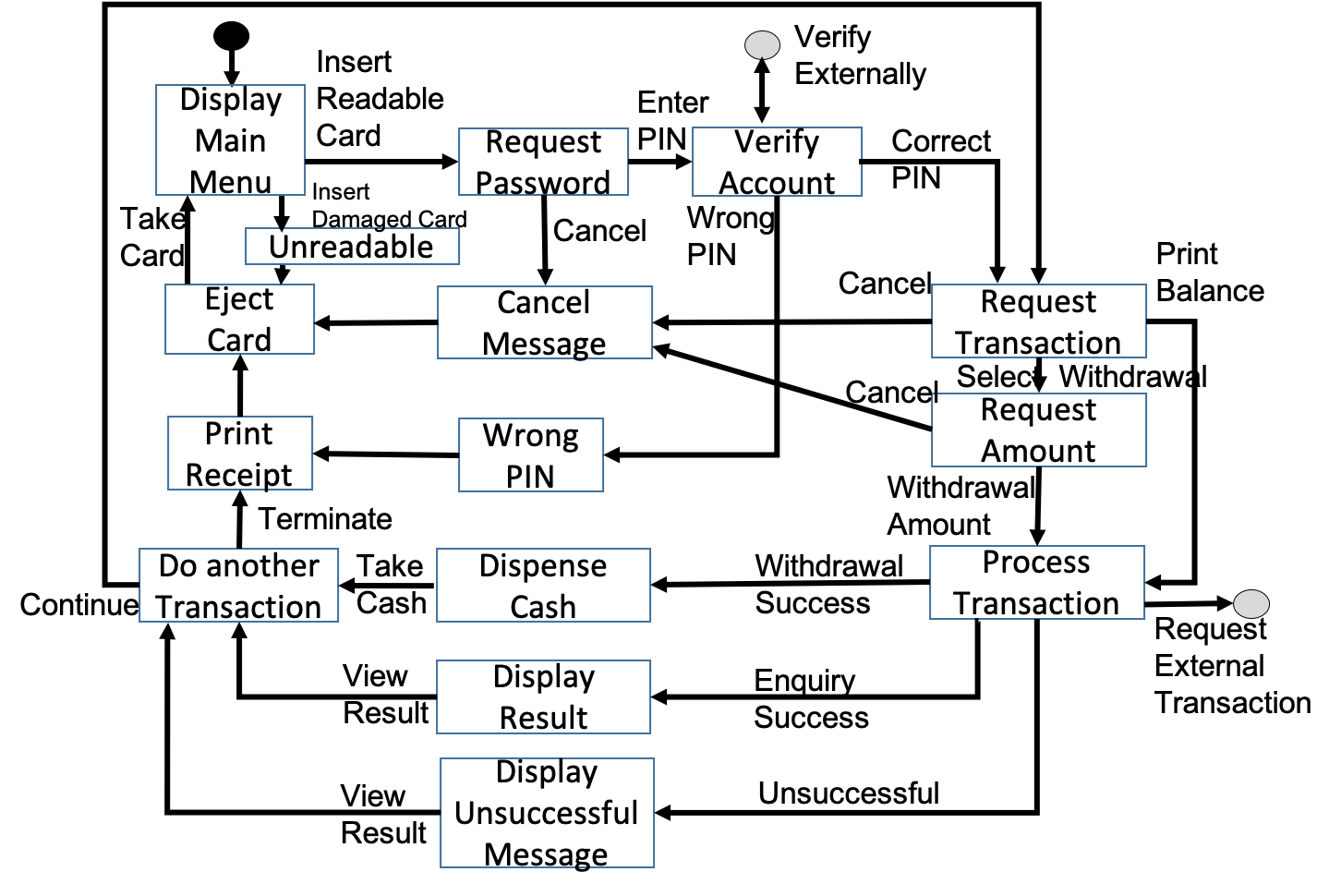}
  \captionof{figure}{ATM State Machine: Known Process Model}
    \label{fig:ATM-sm}
\end{figure}
\begin{figure}
  \includegraphics[height=3.5in, width=3.5in]{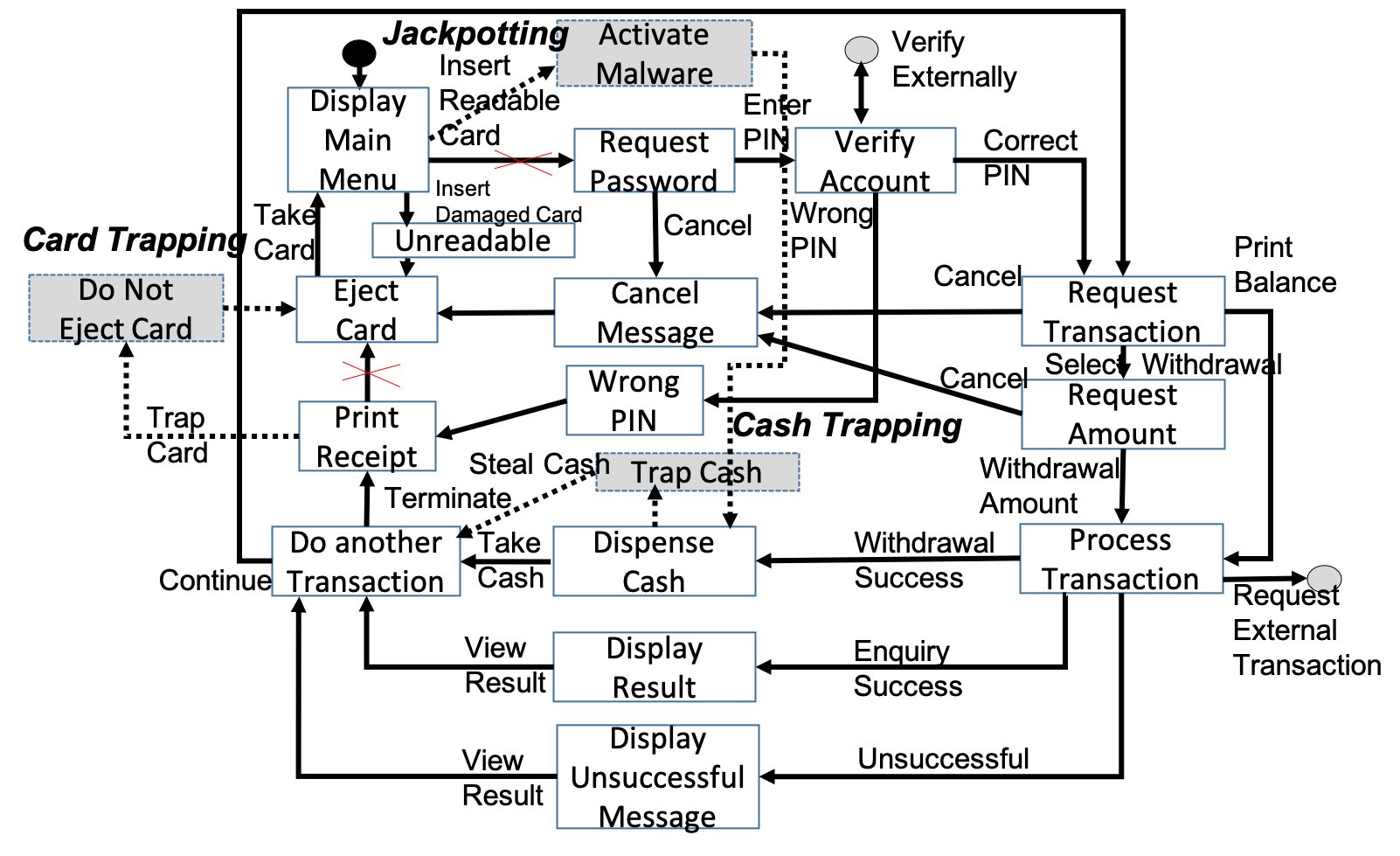}
  \captionof{figure}{ATM State Machine: Ground Truth Model}
    \label{fig:ATM-sm-effects}
\end{figure}

A fairly comprehensive summary of ATM attacks and detection mechanisms
are described by Priesterjahn et al. in
\cite{priesterjahn2015generalized}. The prominent ATM attacks fall in
to the following categories: card skimming, where the information in
the inserted bank card is skimmed for later use; card trapping, where
the bank card is trapped after the transaction is complete using an
extraneous device by the attacker in order to steal it later; PIN
capturing, where the user's PIN number is stolen while being used;
cash trapping, where the cash dispensed is trapped by the attacker
using an extraneous device in order to retrieve it later; brute force
safe opening to steal the cash in ATM; and using malware to attack the
ATM firmware \cite{priesterjahn2015generalized}. There are several
variations of a recent malware attack known as "jackpotting"
\cite{jackpot1}, and all these variations seek to change the firmware
processing of the ATM. Figure \ref{fig:ATM-sm-effects} describes the
reached states and the \textit{ground truth} model of the ATM system
while under two physical attacks: card trapping, and cash trapping, and under jackpotting cyber attack.

When a card trapping attack occurs, the ATM is prevented from going to
the "Eject Card" state by the attacker. Instead, the attacker's device
prevents the card from ejecting (shows as the "Trap Card" transition),
and once the customer leaves, attacker retrieves the card. The ATM
system is unable to observe this card trapping, and this forced
transition is observed incorrectly as a known transition by the ATM
system.  While under the cash trapping attack, once the "Dispense
Cash" state is exited, the system goes to "Trap Cash" state, where the
attacker has inserted a device to prevent dispensing of cash. The
attacker then removes the device, and steals the cash. Since the
\textit{known process} model does not have this forced transition, for
the ATM system, the forced transition is observed as the "Take Cash"
transition, a known transition. Again, the attack goes undetected. In
all these cases, there were no sensors or state variables to detect the
effects of these attacks, and the effects were undetectable using the
\textit{known process} model. In the jackpotting attack, after the
"Insert Readable Card" state, the system transitions to "Activate
Malware" state, and skips the authentication states, and directly
transitions to the "Dispense Cash" state, a transition from an unknown
reached state to a known state.

\subsection{Practical Implications}\label{sec:implications}
We believe the theoretical work described in Section \ref{sec:states} will have implications in several areas.  In this section, we discuss the practical implications  to attack detection, cyber Testing \& Evaluation (T\&E), and system design.

\subsubsection{Attack Detection}
In  Section  \ref{sec:process-model}  we established that  \textit{inconsistent process model} is a control loop effect of several kinds of cyber and physical effects on a cyber-physical system. In Section \ref{sec:process-model}, we formalized the concept of \textit{inconsistent process model} as \textit{Process Model Inconsistency} (PMI).
 In legacy systems, the existing sensors and instrumentation are
 primarily used to support the normal operations, and hence the \textit{known process} model of the
 system. Further, the \textit{known process} model is implemented
 using the controller firmware. Even if we assume the firmware has no
 bugs and accurately implements the known process model, a very
 unlikely situation in most systems, by Theorem
 \ref{theorem:detection}, we could predict it is likely that  some of
 the cyber attacks cannot be detected by just instrumenting the
 firmware, since this will not help make the known process model complete.

 While it is not possible to develop a comprehensive \textit{ground
   truth} process model a priori and identify all the missing state variables
 that need to be monitored for effects, it is important to maintain and augment the \textit{known process} model for
 critical system components when the system is operating, well after the
 design is complete. The possibility of incorrect observation of a
 forced transition as a known transition is a concrete possibility,
 implying if cyber attacks are only reported by observing the effects
 in the firmware of the controller, the attack reports will be
 inaccurate or underreported.  To detect attacks as they happen,
 dynamic behavior of a system needs to be analyzed, including the
 communications to and from the controller. Using threat intelligence
 information and circumstantial evidence may also need be used to
 augment detection capabilities, since in many real systems proving an
 attack has happened with evidence may not be feasible because systems
 are not instrumented to collect relevant evidence. 

 Theorem \ref{theorem:detection} opens the possibility of correctly
 observing a forced state transition, say ($s_a,s_b$). This is a kind
 of anomaly detection, and so some attacks may be correctly detected
 in this way. However, the potential to miss forced transitions points
 to limitations in the effectiveness of anomaly detection, because it
 is quite possible that anomalies may not be observable or correctly
 observed in the \textit{known process} model, leading to higher false
 negatives.

\subsubsection{Cyber T\&E of CPS}
Cyber T\&E concerns with the testing of CPS under various attack scenarios, and cyber modeling \& simulation (cyber M\&S) is a commonly used approach to conduct cyber T\&E \cite{Damodaran:2015:CMS}.  Observing the effects of attacks is a key aspect of cyber T\&E. 
The previous section points to the need to expand instrumentation and
sensors to variables beyond what is needed for system
operations. This need applies to models used in cyber M\&S also. Otherwise, the testers themselves may experience PMI
without being aware of it. 

\subsubsection{System Design}
An improved system design approach would need to focus on detection of
forced transitions, irrespective of whether these transitions lead to
normal or bad states. For example, in the ATM example, if somebody
attempts entering PIN more than 10 times, it is not detectable, though
the entire state machine for entering a wrong PIN is in the known
process model. The known process model may need to be expanded to
include "Count Wrong PIN Attempts" and associated transitions. Another
improvement that could help is the elimination of the default
assignment of states. For example, in the ATM example, once cash is
dispensed, there is no timer to detect whether the cash has left the
dispense tray in a timely manner. The ATM system, after "Dispense
Cash" state, could prevent the default transition to "Another
Transaction," if there was a timer that detects cash has been in the
dispense tray longer than a preset time. The cash trapping attack
might be detected this way.

\section{Conclusions}   \label{sec:conclusions}

 We introduced Process Model Inconsistency (PMI) as an important control loop effect of several types of physical and cyber attacks on Cyber-Physical Systems in Section \ref{sec:process-model}. We showed that it is quite possible to either not observe a PMI effect at all during or after an attack, or come to incorrect conclusions based on the observations of the effects of an attack on the controller or firmware. We illustrated the theoretical results with an example of an Automated Teller Machine (ATM) undergoing two physical and a malware attack.  We also described some practical implications of the limitations on observability  in the areas of attack detection, cyber T\&E, and system design. Evaluating these implications rigorously from a security perspective, and improving the security of new and legacy CPS based on these implications remain to be done. This paper does not address the impact of cyber attacks on liveliness properties, and we hope to address this in future work.



\begin{acks}
We thank Gabriel Pascualy for his work on the elevator. We are grateful to Hasan Cam for comments on an earlier version, and to the anonymous reviewers for their detailed comments.
\end{acks}


\bibliographystyle{ACM-Reference-Format}
\bibliography{pmi-full}

\end{document}